\documentclass[12pt,a4paper]{article}
\usepackage[margin=1in,footskip=0.5in]{geometry}
\usepackage{lscape}
\usepackage{mathtools}
\usepackage{appendix}
\usepackage{xcolor}
\usepackage{amssymb,latexsym}
\usepackage{graphicx}
\usepackage{fancyhdr}
\usepackage{cite}
\usepackage[english]{babel}
\usepackage{nicefrac}
\usepackage{bbm}
\usepackage{amssymb}
\usepackage{tabularx}
\usepackage{epsfig}
\usepackage{amsfonts}
\usepackage{amscd}
\usepackage{amsmath,amsthm}
\usepackage{setspace}
\usepackage[T1]{fontenc}
\RequirePackage{wasysym}
\RequirePackage{setspace}
\usepackage{enumerate}
\usepackage{mathrsfs}
\usepackage[linktocpage, colorlinks=true,citecolor=blue,linkcolor=blue,urlcolor=blue]{hyperref}

\newcommand{\arxiv}[1]{\href{http://arxiv.org/abs/#1}{\texttt{arXiv:#1}}}

\theoremstyle{plain}

\newtheorem{theorem}{Theorem}
\newtheorem{lemma}[theorem]{Lemma}

\newtheorem{proposition}[theorem]{Proposition}
\newtheorem{fact}[theorem]{Fact}

\theoremstyle{definition}

\theoremstyle{remark}

\renewcommand\appendix{\par
	\setcounter{section}{0}
	\setcounter{subsection}{0}
	\setcounter{figure}{0}
	\setcounter{table}{0}
	\renewcommand\thesection{Appendix \Alph{section}}
	\renewcommand\thefigure{\Alph{section}\arabic{figure}}
	\renewcommand\thetable{\Alph{section}\arabic{table}}
}
\newcommand{\be}{\begin{equation}}
\newcommand{\ee}{\end{equation}}
\newcommand{\bea}{\begin{eqnarray}}
\newcommand{\eea}{\end{eqnarray}}
\newcommand{\bfa}{\begin{fact}}
	\newcommand{\efa}{\end{fact}}
\newcommand{\bin}{\begin{inequality}}
	\newcommand{\ein}{\end{inequality}}

\def\l {{\lambda}}

\title{ Fractional Exclusion Statistics\\
	as an Occupancy Process}

\author{Nour-Eddine Fahssi\\
	\small \it Department of Mathematics, FSTM, Hassan Second University\\[-0.8ex]
	\small \it BP 146, Mohammedia,  Morocco.\\
	\small and \\
	\small \it Lab of High Energy Physics, Modeling and Simulation,\\[-0.8ex]
	\small \it Faculty of Science, Mohammed V University at Agdal, \\[-0.8ex]
	\small \it Rabat, Morocco. \\
	\small \href{mailto:noureddine.fahssi@univh2c.ma}{\tt noureddine.fahssi@univh2c.ma}}

\date{}
\begin{document}
	\maketitle

\begin{abstract}
We show the possibility of describing fractional exclusion statistics (FES) as an occupancy process with global and \textit{local} exclusion constraints. More specifically, using combinatorial identities, we show that FES can be viewed as ``ball-in-box" models with appropriate weighting on the set of occupancy configurations (merely represented by a partition of the total number of particles). As a consequence, the following exact statement of the generalized Pauli principle is derived: for an $N$-particles system exhibiting FES of extended parameter \mbox{$g=q/r$} ($q$ and $r$ are co-prime integers such that $0 < q \leq r$), (1)~the allowed occupation number of a state is less than or equal to $r-q+1$ and \emph{not} to $1/g$ whenever $q\neq 1$ and (2)~the global occupancy shape is admissible if the number of states occupied by at least two particles is less than or equal to $(N-1)/r$ ($N \equiv 1 \mod r$). These counting rules allow distinguishing infinitely many families of FES systems depending on the parameter $g$ and the size $N$. 
\medskip

\end{abstract}
%\tableofcontents

\section{Introduction}\label{sec:introduction}

Although, the elementary particles seem to be exhausted by the Pauli classification into bosons and Fermions, topological considerations allow us to generalize the standard statistics in 1D (exchange/anyon statistics~\cite{leinas}) and in 2D (fractional/exclusion statistics~\cite{poly}).  In arbitrary dimensions, Haldane~\cite{Hald}, motivated by the properties of quasi-particles in the fractional quantum Hall effect and in one dimensional inverse-square exchange spin chains, introduced the fractional exclusion statistics (FES)~\cite{Hald}. His proposal is to consider systems with a generalized Pauli blocking interpolating between ``no exclusion'' and ``perfect exclusion''. More explicitly, consider $N$ particles on a lattice of dimension $K$. If we fix the positions of $N-1$ particles, then the remaining single particle can occupy $d_N=K-g(N-1)$ positions. The constant $g$ is a ``statistical interaction'' parameter given by $g=-\Delta d / \Delta N$ where $\Delta d$ is the  change in the dimension of a single particle space and $\Delta N$ is the change in the number of particles. The conventional Bose-Einstein (BES) and Fermi-Dirac (FDS) statistics correspond, respectively, to $g = 0$ and $g = 1$. Now, clear confirmations for FES have been found as well as applications in numerous models of interacting particles; see, e.g.,~\cite{ref2,ref3,ref4,ref5,ref6,anghel,anghel2,CS,Hu} and references therein. In these notes, we designate FES with parameter $g$ by FES$_g$, and term particles obeying FES$_g$ $g$-ons. 

In his seminal paper, Haldane postulated that the full Hilbert space for $g$-ons systems has the size:~\cite{Hald, ref3} \be \label{HW} W_{g}(K,N)= {d_{N}+N-1 \choose N},\ee  where the generic binomial coefficient $\binom{a}{b}$ equals $a!/(b!(a-b)!)$ if $0 \leq b \leq a$, and vanishes otherwise. However, no concrete counting procedure is behind the interpolating formula \eqref{HW} as is the case for conventional statistics. Indeed, BES and FDS are examples of ball-in-box models, i.e., models for random allocations of unlabeled balls/particles in labeled boxes/states, subject to global exclusion constraints. The combinatorial weight --the number of micro-states-- coincides here with the number of ways to distribute the balls among the boxes. In this letter, we want to see if~\eqref{HW} is the combinatorial weight of a ball-in-box model.

First, we must have $N\equiv 1 \pmod r$ so that the dimension $d_N$, and accordingly $W_g(K,N)$, is a whole number. Thus, if $N = r P +1$ for some integer $P$ and $g=q/r$ ($q<r$ are coprime), then $d_{N+r}-d_N=-q$, viz. adding $r$ particles reduces the number of available states by $q$. The weight~\eqref{HW} takes now the form \be \label{HWbis} W_g(K,0)=1, \quad \hbox{and}  \quad W_g(K,N) ={K+(r-q)P \choose r P+1}. \tag{$1'$}\ee Note that $W_g(K,N)=0$ if $P > (K-1)/q$. 

Certainly, dealing with FES using standard techniques of statistical mechanics, allows to describe the thermodynamic properties of $g$-ons~\cite{ref3,ref4,Nayak}, but leads to some inconsistencies when we attempt to find the probabilities for various occupation numbers of a single state~\cite{Nayak}. For instance, Polychronakos~\cite{Poly} proposed a \textit{multiplicative} model which accurately gives back the statistical mechanics of FES in the thermodynamic limit. Multiplicativity here means that, at least for large $K$, the grand partition function is the $K$-th power of a $K$-independent function~\cite{Poly}. However, the price paid for this microscopic realization is the occurrence of negative probabilities. Now, we understand that this problem occurs because Haldane statistics is not multiplicative and, unlike the Pauli principle, the exclusion operates on more than one level. Chaturvedi and Srinivasan~\cite{chat}, and subsequently Murthy and Shankar~\cite{neg}, showed with a remarkable tour de force how negative weights may be avoided for $g=1/2$ (semions) and for $g=1/3$. The more general case of FES$_{1/m}$ was worked out subsequently and explicitly realized in models in one dimension, see~\cite[Chapters 3 \& 4]{ref6}. 

In this letter, we revisit and solve this problem in a closed form when the parameter $g$ is generally any irreducible fraction: $g=q/r$, where $q$ and $r$ are coprime and $0 < q \leq r$.  Our approach is purely combinatorial. It consists of regarding FES as a ball-in-box model, i.e. as an ideal quantum gas, with appropriate weighting on the set of occupancy configurations. To do so, we write the dimension $W_{g}(K,N)$ in the generic form~\eqref{Wbis} below. While calculating the corresponding weights, we deduce the following exclusion rules: An occupancy configuration is allowed if (1) the maximal number of particles that each state can accommodate is $r-q+1$, and not to $g^{-1}$ whenever $q\neq 1$, and (2) the configurations in which the number of states with two or more particles is greater than $(N-1)/r$ are forbidden. This will allow us to distinguish infinitely many families of FES$_g$ systems depending on $g$ and $N$.

%%%%%%%%%%%%%%%%%%%%%%%%%%%%%%%%%%%%%%%%%%%%%%%%%

\section{Multiplicative vs. non-multiplicative statistics}\label{s2}
Let us first set the notation for our combinatorial analysis. The terminology comes mainly from the theory of integer partitions. 

A \emph{partition} of a non-negative integer $N$ is a non-increasing sequence of positive integers whose sum is $N$. To indicate that $\l$ is a partition of $N$, we write $\l \vdash N$ and denote \mbox{$\l=(1^{k_1} 2^{k_2} \ldots N^{k_N})$}, where $\sum_{i=1}^N i k_i = N$ and $k_i$ designates the multiplicity of the part $i$; the sum $\ell(\l)=\sum_{i=1}^N k_i$ is called the \emph{length} of $\l$. The \emph{Ferrers diagram} of $\l$ is a pattern of dots, with the $j$th row having the same number of dots as the $j$th term in $\l$.

Suppose we have $N$ ideal particles to be randomly distributed into $K$ states. In this work, this quantum system will be referred to as a ball-in-box model (indistinguishable balls and labeled boxes). An \emph{occupancy configuration} of the system is said to be of shape $\l=(1^{k_1} 2^{k_2} \ldots N^{k_N}) \vdash N$ if $k_i$ states are occupied by $i$ particles ($i=1,\ldots, N$) and the number of non-vacant states $\ell(\l)$ is less than or equal to $K$. Moreover, if no parts of $\l$ exceed a fixed integer $m$, the corresponding configuration is additionally characterized by $\ell(\l^*) \leq m$, where $\l^*$ stands for the \emph{conjugate} partition  of $\l$, that is, the partition whose Ferrers diagram is obtained from $\l$ by reflection with respect to the diagonal so that rows become columns and columns become rows. For instance, the fermionic configuration is of shape $(1^N)$, characterized by $\ell(\l^*) \leq 1$.

\subsection{A basic formula}
We begin with a simple observation. 
\begin{proposition} \label{prop1} For bosons and fermions, the combinatorial weight can be uniquely written as weighted sums over partitions: \be \label{Wbis} W(K,N)= \sum_{{\l=(1^{k_1} 2^{k_2} \ldots N^{k_N})\; \vdash N}}w(\l) \frac{\ell(\l)!}{k_1! \, k_2! \cdots k_N!} \binom{K}{\ell(\l)}, \ee  where the function $w$ is given by \[w(\l)= \left\{
	\begin{array}{ll}
	1, & \hbox{\emph{for bosons};} \\
	\delta_{\l,(1^N)}, & \hbox{\emph{for fermions}.}
	\end{array}
	\right.\] ($\delta$ is the Kronecker symbol.)     
\end{proposition}  
\begin{proof} We shall apply the generating function method dear to combinatorialists~\cite{wilf}. For bosons, it is well-known that the grand partition function (or the generating function of $W(K,N)$ with respect to $N$) is \[\sum_{N=0}^{\infty} \binom{K+N-1}{N} z^N= (1-z)^{-K}= \left(\sum_{i=0}^\infty z^i\right)^K.\] where $z$ is the fugacity. On the other hand, by the expansion of power series raised to integral powers~\cite[page 823]{stegun}:
	\be \left(\sum_{i=0}^{\infty}a_{i}z^{i} \right)^{K} =
	\sum_{n=0}^{\infty} \left(\sum \frac{K!}{k_{0}! k_{1}! \cdots
		k_{n}!} a_{0}^{k_{0}}a_{1}^{k_{1}} \cdots a_{n}^{k_{n}}\right)
	z^n ,\ee where the inner sum is over the set $\{k_{i} |
	\sum_{i=0}^{n} k_{i}=K \, \mbox{and} \sum_{i=1}^{n}i k_{i}=n \}$,
we write the partition function as	
	\[\sum_{N=0}^{\infty} \binom{K+N-1}{N} z^N = \sum_{N=0}^{\infty}  \left(\sum_{\{k_i\}} \frac{K!}{(K-k_1- \ldots -k_N)! \, k_1! \cdots k_N!}\right) z^N,\] where the inner sum runs over all $N$-tuples $(k_1, \ldots , k_N)$ subject to $k_1+2k_2 + \ldots + N k_N =N $, i.e., over all partitions of $N$. Equating the coefficients on both sides, the bosonic combinatorial weight takes the form \[\sum_{\{k_i\}}\binom{K}{k_1+ \ldots + k_N} \frac{(k_1+ \ldots + k_N)!}{k_1! k_2! \cdots k_N!}=\sum_{{\l\; \vdash N}} \frac{\ell(\l)!}{k_1! \, k_2! \cdots k_N!} \binom{K}{\ell(\l)}.\] Thus, for bosons, $w(\l)=1$. As for fermions, since the partition function is $(1+z)^K$, we find similarly that $w(\l)=1$ if $k_1=N$ and $w(\l)=0$ otherwise.  	
\end{proof}
The expression~\eqref{Wbis} is generic for multiplicative models where the grand partition function is the $K$-th power of an analytic function: $\sum_i a_n z^n$ with positive and $K$-independent Taylor coefficients ($a_0=1$). Reproducing the proof of Proposition~\ref{prop1} yields the weight $w(\lambda) = \prod_{n=1}^N a_n^{k_n}$~\cite{fahssi}. A particularly well-studied example of multiplicative models is the intermediate Gentile statistics of order $G \geq 1$~\cite{gent} for which $a_n=1$ if $n \leq G$ and $a_n=0$ otherwise. So, the one-configuration weight reads here
\be \label{wG} w_G(\l)= \left\{
\begin{array}{ll}
1, & \hbox{if} \quad \ell(\l^*) \leq G; \\
0, & \hbox{otherwise.}
\end{array}
\right.  \ee 

The summands of Eq.~\eqref{Wbis} have the following interpretation: for a fixed configuration of shape $\l$, the factor $\frac{\ell(\l)!}{k_1! \, k_2! \cdots k_N!} \binom{K}{\ell(\l)}$  is the number of ways to choose $\ell(\l)$ non-vacant states out of $K$ ones and arrange $k_i$ states with $i$ particles among them; $i=1, \ldots, N$. The result is then weighted by a non-negative function $w(\l)$. Of course, a configuration $\l$ with $w(\l)=0$ does not contribute to the counting and therefore is not permitted. Thus, the function $w(\lambda)$ encodes the counting rules of the statistics : for bosons, all the configurations contribute equally in the counting, while, for fermions, the Kronecker's delta $\delta_{\l,(1^N)}$ is a manifestation of the Pauli principle.

The FES is not multiplicative in the sense of the definition above. For instance, it is easy to verify that the grand partition function for semions is given by \be \sum_{n\geq0}W_{1/2}(K,N)z^N=\frac{2}{\sqrt{z^2+4}}\left(\frac{z}{2}+\sqrt{1+\frac{z^2}{4}}\right)^{2K}.\ee
As discussed in the introduction, Polychronakos, through a slight modification of $W_g(K,N)$~\cite[Eq.5]{Poly} which leads to the same statistical mechanics, proposed a microscopic realization of FES based on multiplicativity. It can be shown that the Polychronakos modified grand partition function is effectively the $K$-th power of \be {1 \over 4}\left(z+\sqrt{4+z^2}\right)^2= 1+z+\frac{z^2}{2}+\frac{z^3}{8}-\frac{z^5}{128}+\frac{z^7}{1024}-\frac{5
z^9}{32768}+\cdots \ee The Taylor coefficients of this function are not always positive. Therefore, their interpretation as probabilities is problematic~\cite{Nayak}.

Interestingly, the dimension~\eqref{HWbis} can be cast in the generic form~\eqref{Wbis} with a well-defined and positive weighting function $w(\l)$. It is the latter that must be perceived as probability of (global) occupation. More precisely, we show

\begin{theorem} \label{HWweight} For $g=q/r$ and $N \equiv 1 \pmod{r}$ , the number of micro-states \emph{\eqref{HWbis}} can uniquely be written in the form~\eqref{Wbis}, where
 \be  w_g(\l) = \label{HWw} \binom{(N-1)/r}{\ell(\l)-k_1} \binom{\ell(\l)}{k_1}^{\!\!-1}  \,\prod_{j=0}^{r-q} \binom{r-q}{j}^{k_{j+1}}. \ee In particular, $w_g(\l)=0$ if $\ell(\l^*) > r-q+1$. 
\end{theorem} \noindent We postpone the proof of our main result to Sub-section.~\ref{proof}.
 
In the case with $g=1/2$, the weight~\eqref{HWw} reads
\[w_{1/2}(\l)=\binom{(N-1)/2}{k_2}\binom{k_1+k_2}{k_1}^{-1}, \qquad \l=(1^{k_1}2^{k_2}) \vdash N\] which is exactly the formula derived by Chaturvedi and Srinivasan in their microscopic interpretation of semion statistics~\cite{chat}. For $g=1/3$,   \[w_{1/3}(\l)= \binom{(N-1)/3}{k_2+k_3}\binom{k_1+k_2+k_3}{k_1}^{-1} 2^{k_2}, \qquad \l=(1^{k_1}2^{k_2}3^{k_3}) \vdash N, \] formula obtained by Murthy and Shankar using a different approach~\cite{neg}.

From the expression~\eqref{HWw}, we underline the following features:   \begin{enumerate}
	\item[(1)] the weights $w_g (\l)$ are fractional and non-negative definite,
	\item[(2)] the weights $w_g(\l)$ depend only upon $P \coloneqq (N-1)/r$ and the difference $r-q$,
	\item[(3)] the allowed occupation number for a single-state does not exceed $r-q+1$ and \emph{not} $1/g$ whenever $q\ne 1$. We recall, however, that the average occupation number is less than $1/g \leq r-q+1$~\cite{ref3},
	\item[(4)] Since the binomial coefficient $\binom{P}{\ell(\l)-k_1}$ in~\eqref{HWw} vanishes if $P < \sum_{i=2}^m k_i$ , the corresponding configuration does not contribute to the counting.
\end{enumerate} The last observation is crucial. It stipulates that a necessary condition for permissible configurations is that the number of states occupied by two particles or more is less than or equal to $P$. In other words, the Ferrers diagram of the partition $(2^{k_{2}} \ldots (r-q+1)^{k_{r-q+1}})$, extracted from $\l$, must fit inside the rectangle $[P \times (r-q+1)]$.

Let us incorporate the above-formulated rules as follows: 
\medskip

\noindent \textbf{Generalized Exclusion Principle}.
	\emph{A configuration of shape $ \l \vdash n$ is admissible if and only if the following constraints are fulfilled:}
\begin{enumerate}[leftmargin=5cm]
	\item[$C_1$\; \emph{:}] \quad $\ell(\l) \leq K$ \; \quad (\emph{by definition}),
	\item[$C_2$\; \emph{:}] \quad  $\ell(\l^*) \leq r-q+1$  \quad (\emph{at most $r-q+1$ particles per state}),
	\item[$C_3$\; \emph{:}]  \quad $\displaystyle \sum_{i=2}^m k_i \leq \frac{N-1}{r} \leq  \frac{K-1}{q}$ \quad ($w_g(\l) \neq 0$ \emph{and} $d_N \geq 1$). \end{enumerate}

Therefore, the exclusion operates not only on the ``microscopic'' level (condition $C_2$), but also on the ``macroscopic'' level (condition $C_3$). To illustrate, we implement this in two specific examples. First, let, say, $g=1/3$ and $N=10$. Here the maximal allowed occupancy of a state is $3$. By the constraint $C_2$, 14 configurations may contribute (depending on $K \geq 4$), among which the configurations $(1^2 2^4)$, $(2^5)$, $(1 2^3 3)$ and $(2^2 3^2)$  are forbidden by the constraint $C_3$:
\begin{center}\qquad \quad  \includegraphics[width=6cm]{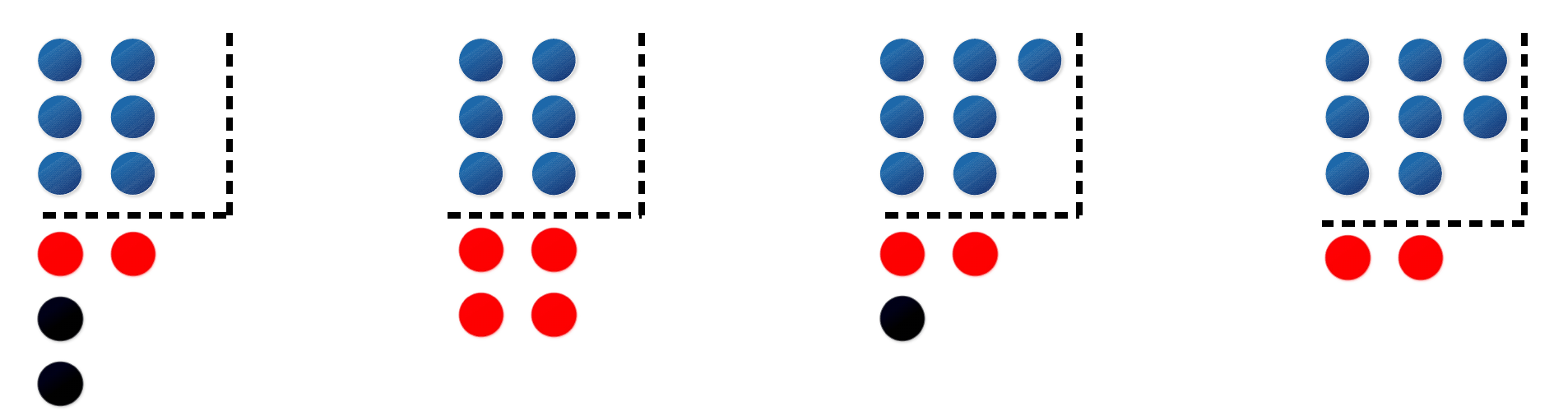}\end{center}
As a second example, take $g=3/5$ and $N=16$. Here the maximal allowed occupancy is again $3$. Among the 231 partitions of 16, only 10 may contribute to the total weight: 
\begin{center}
{\footnotesize	\begin{tabular}
			{c||c|c|c|c|c|c|c|c|c|c}
			$\lambda$ & $(1^{16})$ & $(1^{14} 2)$ & $(1^{12} 2^2)$ & $(1^{10} 2^3)$ & $(1^{13} 3)$ & $(1^{11} 2\,3)$ & $(1^9 2^2 3)$ & $(1^{10} 3^2)$ & $(1^7 3^3)$&$(1^8 2\,3^2)$\\
			%&&&&&&&&&&\\
			\hline
			%&&&&&&&&&&\\
			$w(\lambda)$ & $1$& $2/5$ & $12/91$ & $4/143$ & $3/14$ & $1/13$ & $1/55$ & $1/22$ & $1/120$& $2/165$  \\
	\end{tabular},}
\end{center}
each of which contributes only if its length is less than or equal to $K \geq 10$. We check readily that summing the contributions of all allowed configurations yields $\binom{K+6}{16}=W_{3/5}(k,16)$.

It is worth noting that, in view of the constraints $C_2$ and $C_3$, we can distinguish infinitely many families of $g$-ons systems according to $P=(N-1)/r$ and the difference $ r-q $. Indeed, representing an $N$-particle system fulfilling FES$_g$ by the pair $(N,g=q/r)$, two systems $(N,g=q/r)$ and $(N',g')$ are subject to the same exclusion rules if there exist an integer $j >0$ not a multiple of $r-q$ such that
\be   g'=\frac{j}{r-q+j}, \quad \hbox{and} \quad
\frac{N'-1}{r-q+j}=\frac{N-1}{r}.\ee
The semions, for example, belong to the family with $g=j/(j+1)$, the \emph{semionic family}. Clearly, the Bose and Fermi statistics are recovered in the limits $j=0$ and $j \to \infty$ respectively.

\begin{proposition} \label{prop3}
	For $N \leq K$, the number of permissible configurations is \be \label{conf} \binom{(N-1)/r +r-q}{r-q}. \ee
\end{proposition} 
\begin{proof} Clearly, when $N \leq K$ the condition $C_1$ and the inequality in the right of the constraint $C_3$ are satisfied. Thus, a configuration $\l$ is likely if and only if the inequality in the left of the condition $C_3$ holds true. Therefore, the number of allowed configurations is the number of solutions of $k_2 + k_3 + \cdots + k_m \leq (N-1)/r$ in nonnegative integers. The result follows from the known fact that the number of solutions of $x_1 + x_2 + \cdots + x_k \leq p$ is given by $\binom{p+k}{k}$ (cf.~\cite[p.103]{vanLint}). \end{proof}
%==============================
By way of comparison, the exact exclusion rules for the Gentile statistics are, in addition to $C_1$, $\ell(\l^*) \leq G$ and $N \leq G K$. Thus, the number of permitted configurations is simply that of the partitions of $N$ with no more than $K$ parts; no part exceeding $G$. This number is the coefficient of $q^N$ in the Gaussian polynomial $\left[ \!{K+G \atop K} \!\right]_q$~\cite[Chap.3]{Andr}. When $N\leq K$, this reduces to the number of partitions with largest part not exceeding $G$. We also emphasize that if $G=r-q+1$, then $W_G(K,N)$  majorizes $W_g(K,N)$ since the exclusion principle of FES is more restrictive.

\subsubsection{Proof of Theorem~\ref{HWweight}\label{s4}}
\label{proof}	
To prove Theorem~\ref{HWweight}, we need the following identity:
\begin{lemma} Let $P$, $n$ and $k$ be positive integers. Then
	\be \label{HF} \binom{k P}{n}= \sum_{\{l_i\}} \frac{P!}{l_1! \cdots l_k!(P-l_1- \cdots -l_k)!}
	\prod_{i=1}^k \binom{k}{i}^{l_i}, \ee where the sum runs over all $k$-tuples $(l_1, \ldots , l_k)$ subject to the constraint $l_1+2l_2 + \ldots +k l_k =n $. \end{lemma}

\begin{proof} We shall again use the technique of generating function. Let $z$ be an indeterminate. On one hand, we have by application of the binomial theorem  \be \label{ex1} (1+z)^{kP}=\sum_{n=0}^{kP}\binom{kP}{n} z^n,\ee and, on the other hand,  by the well-known multinomial theorem:  \bea \nonumber (1+z)^{kP}&=&\left((1+z)^k\right)^P = \left(\sum_{i=0}^k \binom{k}{i}z^i\right)^P = \sum_{l_0 + l_1 + \cdots l_k=P} \frac{P!}{l_0! \, l_1! \, \cdots l_k!}\prod_{i=0}^k \left( \binom{k}{i}z^i \right)^{l_i} \\
	\label{ex2} &=& \sum_{(l_1 , \ldots , l_k)} \left(\frac{P!}{ l_1! \, \cdots l_k! \, (P-l_1-l_2-\cdots l_k)! }\prod_{i=0}^k \binom{k}{i}^{l_i}\right) \; z^{l_1+2l_2+\cdots +k l_k}. \eea   The identity~\eqref{HF} follows by equating the coefficients of $z^n$ in the two expansions~\eqref{ex1} and~\eqref{ex2}.\end{proof}  \begin{proof}[Proof of Theorem~\ref{HWweight}]  Inserting the weight $w_g(\l)$, the RHS of \eqref{Wbis} can be displayed as  \[ \sum_{{\l \vdash N \atop \ell(\l^*) \leq r-q+1}} \left(
	\frac{P!}{k_2! \cdots k_{r-q+1}!(P-k_2-\cdots-k_{r-q+1})!} \prod_{i=1}^{r-q}
	{{r-q}\choose i}^{k_{i+1}} \right) {{K}\choose \ell(\l)}, \] where $P=(N-1)/r$. Taking into account that $\ell(\l)=\sum_{i=1}^{r-q}k_{i}=N-\sum_{i=1}^{r-q}ik_{i+1}$
	and putting $s=\sum_{i=1}^{r-q}ik_{i+1}$ (the integer $s$ ranges
	from 0 to $(r-q)P$ since $k_{r-q+1} \leq P$), we
	re-express the last formula as a double sum:
	\be \label{A2} \sum_{s=0}^{(r-q)P} \left( \sum_{\sum_{i=1}^{r-q}ik_{i+1}=s}\frac{P!}{k_2! \cdots k_{r-q+1}!(P-k_2-\cdots-k_{r-q+1})!} \prod_{i=1}^{r-q} {{r-q}\choose i}^{k_{i+1}}\right) {{K}\choose N-s}.\ee  Now we make the change of summation indices $l_i=k_{i+1}$ to write the inner sum as the RHS of formula~\eqref{HF}:
	\be \label{A3} \sum_{\sum_{i=1}^{r-q}i l_{i}=s}\frac{P!}{l_1! \cdots
		l_{r-q}!(P-l_1-\cdots-l_{r-q})!} \prod_{i=1}^{r-q} {{r-q}\choose i}^{l_{i}} = {{(r-q)P}\choose s}.
	\ee
	We deduce finally that the RHS of Eq.~\eqref{Wbis} reads \be \label{Vander} \sum_{s=0}^{(r-q)P}{{(r-q)P}\choose s}{{K}\choose N-s}={{K+(r-q)P}\choose N}=W_g(K,N), \ee where, to obtain the last equality, we employed the well-known Vandermonde's formula for binomial coefficients~\cite{vander}.  \end{proof} 

We stress, finally, that for $g>1$ ($r<q$), one may follow the proof above to check that $W_g(K,N)$ can as well be formally written in the form~\eqref{Wbis}, but the constraint of maximal occupancy became relaxed and the weights \emph{inevitably negative} for some configurations. Indeed, in this case, the weights are not positive definite since $\binom{r-q}{i}<0$ for odd $i$.
%=========================
%%%%%%%%%%%%%%%%%%%%%%%%%%%%%%%%%%%%%%%%%%%%%%%%%%%%%%%%%%
%%%%%%%%%%%%%%%%%%%%%%%%%%%%%%%%%%%%%%%%%%%%%%%%%%%%%%%%%%%

\end{document}